\documentclass{stacs_proc}
\usepackage{graphicx}
\usepackage{pstricks}
\usepackage{pst-node}
\usepackage{pst-coil}
\usepackage{amsmath,amsthm}
\usepackage{amsmath}
\usepackage{amssymb}
\usepackage{amsfonts}
\usepackage{epsfig}
\usepackage{verbatim}

\theoremstyle{plain}

\psset{unit=1pt}

\newlength{\alginputwidth}
\newlength{\algboxwidth}
\newcommand{\alginput}[1]{\makebox[1.5cm][l]{ {\sc Input:}} \parbox[t]{\alginputwidth}{{\it #1}}}
\newcommand{\algoutput}[1]{\makebox[1.5cm][l]{ {\sc Output:}} \parbox[t]{\alginputwidth}{{\it #1}}}
\newcommand{\algtitle}[1]{\underline{Algorithm \ {\bf #1}} \vspace*{1mm}\\}

\newsavebox{\algbox}
\newsavebox{\captionbox}
\newenvironment{algorthm}[2]%
    {
	\setlength{\algboxwidth}{\columnwidth}
	\addtolength{\algboxwidth}{-\columnsep}
	\addtolength{\algboxwidth}{-1mm}
	\setlength{\alginputwidth}{\algboxwidth}
	\addtolength{\alginputwidth}{-1.7cm}
	\begin{figure}[tb]
	    \vspace*{2mm}
	    \centering
	    \begin{lrbox}{\captionbox}
		\begin{minipage}[b]{\algboxwidth}
		    \centering
		    \caption{#1}
		    \label{#2}
		\end{minipage}
	    \end{lrbox}
	    \begin{lrbox}{\algbox}
		\begin{minipage}[b]{\algboxwidth}
		    \footnotesize
		    \vspace*{2mm}
    } 
    {
		    \vspace*{0.2mm}
	       \end{minipage}
	    \end{lrbox}
	    \fbox{\usebox{\algbox}\hspace*{1mm}}
	    \usebox{\captionbox}
	    \vspace*{-4mm}
	\end{figure}
    }
\newsavebox{\algcodebox}
\newenvironment{codeblock}%
    {
	\begin{enumerate}
	    \setlength{\itemsep}{2pt}
	    \setlength{\parsep}{0pt}
	    \setlength{\topsep}{0pt}
	    \setlength{\parskip}{0pt}
	    \setlength{\partopsep}{0pt}
    } 
    {\end{enumerate}}
\newcommand{\step}{\item}

\begin{document}

\title{On Geometric Spanners of Euclidean and Unit Disk Graphs}
\author[]{I. Kanj}{Iyad A. Kanj}
\address[]{DePaul University, Chicago, IL 60604, USA.}
\thanks{The work of the first author was supported in part by a DePaul
University Competitive Research Grant.}

\author[]{L. Perkovi\'{c}}{Ljubomir Perkovi\'{c}}
\email{{ikanj,lperkovic}@cs.depaul.edu}


\begin{abstract}
  \noindent We consider the problem of constructing bounded-degree planar
geometric spanners of Euclidean and unit-disk graphs. It is well known that
the Delaunay subgraph is a planar geometric spanner with stretch
factor $C_{del}\approx 2.42$; however, its degree may not be
bounded. Our first result is a very simple linear time
algorithm for constructing a subgraph of the Delaunay graph with
stretch factor $\rho =1+2\pi(k\cos{\frac{\pi}{k}})^{-1}$ and degree
bounded by $k$, for any integer parameter $k\geq 14$. This result
immediately implies an algorithm for constructing a planar geometric
spanner of a Euclidean graph with stretch factor $\rho \cdot
C_{del}$ and degree bounded by $k$, for any integer parameter $k\geq
14$. Moreover, the resulting spanner contains a Euclidean Minimum Spanning Tree
(EMST) as a subgraph. Our second contribution lies in developing the structural
results necessary to transfer our analysis and algorithm from
Euclidean graphs to unit disk graphs, the usual model for wireless ad-hoc
networks. We obtain a very simple distributed, {\em strictly-localized}
algorithm that, given a unit disk graph embedded in the plane, constructs a 
geometric spanner with the above stretch factor and degree bound, and also
containing an EMST as a subgraph.  The obtained results
dramatically improve the previous results in all aspects, as shown
in the paper. 
\end{abstract}

\maketitle

\stacsheading{2008}{409-420}{Bordeaux}
\firstpageno{409}

\section*{Introduction} \label{intro}
Given a set of points $P$ in the plane, the Euclidean graph $E$ on
$P$ is defined to be the complete graph whose vertex-set is $P$.
Each edge $AB$ connecting points $A$ and $B$ is assumed to be
embedded in the plane as the straight line segment $AB$; we define
its cost to be the Euclidean distance $|AB|$. We
define the unit disk graph $U$ to be the subgraph of $E$ consisting
of all edges $AB$ with $|AB| \leq 1$.

Let $G$ be a subgraph of $E$. The cost of a simple path $A = M_0,
M_1, ..., M_r = B$ in $G$ is $\sum_{j=0}^{r-1} |M_jM_{j+1}|.$ Among
all paths between $A$ and $B$ in $G$, a path with the smallest cost
is defined to be a {\em smallest cost path} and we denote its cost
as $c_G(A,B)$. A spanning subgraph $H$ of $G$ is said to be a {\em geometric
spanner} of $G$ if there is a constant $\rho$ such that for every
two points $A,B \in G$ we have: $c_H(A,B) \leq \rho \cdot c_G(A,B)$.
The constant $\rho$ is called the {\em stretch factor} of $H$ (with
respect to the underlying graph $G$).

The problem of constructing geometric spanners of Euclidean graphs has
recently received a lot of attention due to its applications in
computational geometry, wireless computing, and computer graphics
(see, for example, the recent book~\cite{spannerbook} for a survey on
geometric spanners and their applications in networks). Dobkin et
al.~\cite{dobkin} showed that the Delaunay graph is a planar
geometric spanner of the Euclidean graph with stretch factor
$(1+\sqrt{5})\pi/2 \approx 5.08$. This ratio was improved by
Keil et al~\cite{keil} to $C_{del}=2\pi/(3\cos{(\pi/6)}) \approx
2.42$, which currently stands as the best upper bound on the stretch
factor of the Delaunay graph. Many researchers believe, however, that
the lower bound of $\pi/2$ shown in~\cite{chew} is also an upper
bound on the stretch factor of the Delaunay graph. While
Delaunay graphs are good planar geometric spanners of Euclidean
graphs, they may have unbounded degree. 

Other geometric (sparse) spanners were also proposed in the literature
including the Yao graphs~\cite{yao}, the $\Theta$-graphs~\cite{keil},
and many others (see~\cite{spannerbook}). However, most of these proposed
spanners either do not guarantee planarity, or do not guarantee bounded
degree.

Bose et al.~\cite{boseesa,bosealgorithmica} were the first to show
how to extract a subgraph of the Delaunay graph that is a planar
geometric spanner of the Euclidean graph with
stretch factor $\approx 10.02$ and degree bounded by $27$. In the
context of unit disk graphs, Li et
al.~\cite{iitunbounded1,iitunbounded} gave a distributed algorithm
that constructs a planar geometric spanner of a unit disk graph with
stretch factor $C_{del}$; however, the spanner constructed can have
unbounded degree. Wang and Li~\cite{iitbounded1,iitbounded} then
showed how to construct a bounded-degree planar spanner of a unit
disk graph with stretch factor $max\{\pi/2, 1+
\pi\sin{(\alpha/2)}\}\cdot C_{del}$ and degree bounded by $19 +
2\pi/\alpha$, where $0< \alpha < 2\pi/3$ is a parameter. Very
recently, Bose et. al~\cite{bose1} improved the earlier result in
~\cite{boseesa,bosealgorithmica} and showed how to construct a
subgraph of the Delaunay graph that is a geometric spanner of the
Euclidean graph with stretch factor: $max\{\pi/2, 1+
\pi\sin{(\alpha/2)}\}\cdot C_{del}$ if $\alpha < \pi/2$ and $(1 +
2\sqrt{3} + 3\pi/2 + \pi\sin{(\pi/12)})\cdot C_{del}$ when $\pi/2
\leq \alpha \leq 2\pi/3$, and whose degree is bounded by $14 +
2\pi/\alpha$. Bose et al. then applied their construction to obtain
a planar geometric spanner of a unit disk graph with stretch factor
$max\{\pi/2, 1+ \pi\sin{(\alpha/2)}\}\cdot C_{del}$ and degree
bounded by $14+ 2\pi/\alpha$, for any $0 < \alpha \leq \pi/3$. This
was the best bound on the stretch factor and the degree.

We have two new results in this paper. We develop structural results about
Delaunay graphs that allow us to present a very simple linear-time algorithm
that, given a Delaunay graph, constructs a subgraph of the Delaunay graph
with stretch factor $1 + 2\pi(k \cos{(\pi/k)})^{-1}$ (with respect
to the Delaunay graph) and degree at most $k$, for any integer
parameter $k\geq 14$. This result immediately implies an $O(n\lg{n})$
algorithm for constructing a planar geometric spanner of a Euclidean graph
with stretch factor of $(1 + 2\pi(k \cos{(\pi/k)})^{-1})\cdot
C_{del}$ and degree at most $k$, for any integer parameter $k \geq
14$ ($n$ is the number of vertices in the graph). We then translate our
work to unit disk graphs and present our second result: a
very simple and {\em strictly-localized} distributed algorithm that, given a
unit-disk graph embedded in the plane, constructs a planar geometric spanner
of the unit disk graph with stretch factor $(1 + 2\pi(k
\cos{(\pi/k)})^{-1})\cdot C_{del}$ and degree bounded by $k$, for
any integer parameter $k \geq 14$. This efficient distributed algorithm
exchanges no more than $O(n)$ messages in total, and runs in $O(\Delta
\lg{\Delta})$ local time at a node of degree $\Delta$. We show that both
spanners include a Euclidean Minimum Spanning Tree as a subgraph.

Both algorithms significantly improve previous results (described above)
in terms of the stretch factor and the degree bound. To show this, we compare
our results with previous results in more detail. For a degree bound $k=14$,
our result on Euclidean graphs imply a bound of at most $3.54$ on the stretch
factor. As the degree bound $k$ approaches $\infty$, our bound on the
stretch factor approaches $C_{del} \approx 2.42$. The very recent
results of Bose et al.~\cite{bose1} achieve a lowest degree bound of
$17$, and that corresponds to a bound on the stretch factor of at
least $23$. If Bose et al.~\cite{bose1} allow the degree bound to be
arbitrarily large (i.e., to approach $\infty$), their bound on the
stretch factor approaches $(\pi/2) \cdot C_{del} > 3.75$. Our stretch factor
and degree bounds for unit disk graphs are the same as our results
for Euclidean graphs. The smallest
degree bound derived by Bose et al.~\cite{bose1} is 20, and that
corresponds to a stretch factor of at least 6.19. If Bose et
al.~\cite{bose1} allow the degree bound to be arbitrarily large,
then their bound on the stretch factor approaches $(\pi/2) \cdot
C_{del} > 3.75$. On the other hand, the smallest degree bound derived in
Wang et al.~\cite{iitbounded1,iitbounded} is 25, and that
corresponds to a bound of 6.19 on the stretch factor. If Wang et
al.~\cite{iitbounded1,iitbounded} allow the degree bound to be
arbitrarily large, then their bound on the stretch factor approaches
$(\pi/2) \cdot C_{del} > 3.75$. Therefore, even the worst bound of
at most 3.54 on the stretch factor corresponding to our lowest bound
on the degree $k=14$, beats the best bound on the stretch factor of
at least 3.75 corresponding to arbitrarily large degree in both Bose
et al.~\cite{bose1} and Wang et al.~\cite{iitbounded1,iitbounded}!

\section{Definitions and Background}
\label{background}

We start with the following well known observation:
\begin{observation} \label{optimal} A subgraph $H$ of graph $G$ has
stretch factor $\rho$ if and only if for every edge $XY \in G$: the
length of a shortest path in $H$ from $X$ to $Y$ is at most $\rho
\cdot |XY|$.
\end{observation}
For three non-collinear points $X$, $Y$, $Z$ in the plane we denote
by $\bigcirc{XYZ}$ the circumscribed circle of triangle
$\triangle{XYZ}$. A {\em Delaunay triangulation} of a set of points $P$ in
the plane is a triangulation of $P$ in which the circumscribed circle of every
triangle contains no point of $P$ in its interior. It is well known
that if the points in $P$ are {\em in general position} (i.e., no
four points in $P$ are cocircular) then the Delaunay triangulation
of $P$ is unique~\cite{book}. In this paper---as in most papers in
the literature---we shall assume that the points in $P$ are in
general position; otherwise, the input can be slightly perturbed so
that this condition is satisfied. The {\em Delaunay graph} of $P$ is
defined as the plane graph whose point-set is $P$ and whose edges
are the edges of the Delaunay triangulation of $P$. An alternative definition
that we end up using is:
\begin{definition}
\label{DelaunayDef}
An edge $XY$ is in the Delaunay graph of $P$ if and only if there exists a
circle through points $X$ and $Y$ whose interior contains no point in $P$.
\end{definition}
It is well known
that the Delaunay graph of a set of points $P$ is a spanning
subgraph of the Euclidean graph defined on $P$ (i.e., the complete
graph on point-set $P$) whose stretch factor is bounded by $C_{del}
= 4\sqrt{3} \pi /9 \approx 2.42$~\cite{keil}.

Given integer parameter $k > 6$, the {\em Yao subgraph}~\cite{yao}
of a plane graph $G$ is constructed by performing the
following {\em Yao step} at every point $M$ of $G$: place $k$ equally-separated
rays out of $M$ (arbitrarily defined), thus creating $k$ closed cones of
size $2\pi/k$ each, and choose the shortest edge in $G$ out of $M$ (if
any) in each cone. The Yao subraph consists of edges in $G$ chosen by
{\em either} endpoint. Note that the degree of a point in the Yao subgraph
of $G$ may be unbounded.

Two edges $MX$, $MY$ incident on a point $M$ in a graph $G$ are said
to be {\em consecutive} if one of the angular sectors determined by
$MX$ and $MY$ contains no neighbors of $M$.

\section{Bounded Degree Spanners of Delaunay Graphs}
\label{euclidean} 

Let $P$ be a set of points in the plane and let $E$
be the complete, Euclidean graph defined on point-set $P$. Let $G$ be the
Delaunay graph of $P$. This section is devoted to proving the following
theorem:

\begin{theorem}
\label{maintheorem} For every integer $k \geq 14$, there exists a
subgraph $G'$ of $G$ such that $G'$ has maximum degree $k$ and
stretch factor $1+2\pi(k\cos{\frac{\pi}{k}})^{-1}$.
\end{theorem}

A linear time algorithm that computes $G'$ from $G$ is the key
component of our proof. This very simple algorithm essentially
performs a {\em modified Yao step} (see
Section~\ref{modifiedyaoalgo}) and selects up to $k$ edges out of
every point of $G$. $G'$ is simply the spanning subgraph of $G$ consisting
of edges chosen by {\em both} endpoints.

In order to describe the modified Yao step, we must first develop a
better understanding of the structure of the Delaunay graph $G$. Let
$CA$ and $CB$ be edges incident on point $C$ in $G$ such that
$\angle{BCA} \leq 2\pi/k$ and $CA$ is the shortest edge within the
angular sector $\angle{BCA}$. We will show how the above theorem easily
follows if, for every such pair of edges $CA$ and $CB$:
\begin{enumerate}

\item we show that there exists a path $p$ from $A$ to $B$ in $G$ of length $|p|$,
such that:\\ $|CA|$ $+$ $|p|$ $\leq (1+2\pi(k\cos{\frac{\pi}{k}})^{-1})|CB|$, and

\item we modify the standard Yao step to include the edges of this
path in $G'$, in addition to including the edges picked by the
standard Yao step but without increasing the number of edges chosen at each
point beyond $k$.
\end{enumerate}
This will ensure that: for any edge $CB \in G$ that is not included in $G'$
by the modified Yao step, there is a path from $C$ to $B$ in $G'$, whose
edges are all included in $G'$ by the modified Yao step, and whose cost is
at most $(1+2\pi(k\cos{\frac{\pi}{k}})^{-1})|CB|$. In the lemma below, we prove
the existence of this path and show some properties satisfied by edges of this
path; we will then modify the standard Yao step to include edges satisfying
these properties.
\begin{lemma}
\label{canpath} Let $k \geq 14$ be an integer, and let $CA$ and $CB$
be edges in $G$ such that $\angle{BCA} \leq 2\pi/k$ and $CA$ is the
shortest edge in the angular sector $\angle{BCA}$. There exists a
path $p:$~$A = M_0, M_1, ..., M_r = B$ in $G$ such that:

\begin{itemize}
\item[(i)] $|CA|+\sum_{i=0}^{r-1}|M_iM_{i+1}| \leq
(1+2\pi(k\cos{\frac{\pi}{k}})^{-1})|CB|$.
\item[(ii)] There is no edge in $G$ between any pair $M_i$ and $M_j$ lying in the closed region delimited by
$CA$, $CB$ and the edges of $p$, for any $i$ and $j$
satisfying $0 \leq i < j-1 \leq r$.
\item[(iii)] $\angle{M_{i-1}M_iM_{i+1}} > ({{k - 2} \over k})\pi$, for $i = 1,\cdots,
r-1$.
\item[(iv)] $\angle{CAM_1} \geq {\pi \over 2} - {\frac{\pi}{k}}$.
\end{itemize}
\end{lemma}
We break down the proof of the above lemma into two cases: when
$\triangle{ABC}$ contains no point of $G$ in its interior, and when there
 are points of $G$ inside $\triangle{ABC}$. We define some additional notation
and terminology first. We define the circle $(O) = \bigcirc{ABC}$ with center
$O$, and set $\Theta = \angle{BCA}$. Note that $\angle{AOB} =2\Theta \leq
4\pi/k$. We will use $\stackrel{\frown}{AB}$ to denote the arc of $(O)$
determined by points $A$ and $B$ and facing $\angle{AOB}$. We will make use
of the following easily verified Delaunay graph property:
\begin{proposition}
\label{Euclidproperty} If $CA$ and $CB$ are edges of $G$ then the region inside $(O)$
subtended by chord $CA$ and away from $B$ and the region inside $(O)$ subtended
by chord $CB$ and away from $A$ contain no points.
\end{proposition}

\subsection{The Outward Path}
\label{outpath} We consider first the case when no points of $G$ are
inside $\triangle{ABC}$. Since both $CA$ and $CB$ are edges in $G$ and by
Proposition~\ref{Euclidproperty}, the region of $(O)$ subtended by chord $AB$
closer to $C$ has no points of $G$ in its interior. Keil and Gutwin~\cite{keil}
showed that, in this case, there exists a path between $A$ and $B$ in $G$
inside the region of $(O)$ subtended by chord $AB$ away from $C$, whose
length is bounded by the length of $\stackrel{\frown}{AB}$
(see Lemma~1 in~\cite{keil}). We find it convenient to use a recursive
definition of their path (for more details, we refer the reader
to~\cite{keil}):
\begin{enumerate}
\item {\bf Base case:} If $AB \in G$, the path consists of edge $AB$.

\item {\bf Recursive step:} Otherwise, a point must reside in the
region of $(O)$ subtended by chord $AB$ and away from $C$.
Let $T$ be such a point with the property that the region of
$\bigcirc{ATB}$ subtended by chord $AB$ closer to $T$ is empty.
We call $T$ an {\em intermediate point} with respect to the pair of
points $(A,B)$. Let $(O_1)$ be the circle passing through $A$ and
$T$ whose center $O_1$ lies on segment $AO$ and let $(O_2)$ be the
circle passing through $B$ and $T$ whose center $O_2$ lies on
segment $BO$. Then both $(O_1)$ and $(O_2)$ lie inside $(O)$, and
$\angle{AO_1T}$ and $\angle{TO_2B}$ are both less than $\angle{AOB}
\leq 4\pi/k$. Moreover, the region of $(O_1)$ subtended by chord
$AT$ that contains $O_1$ is empty, and the region of $(O_2)$
subtended by chord $BT$ and containing $O_2$ is empty. Therefore, we
can recursively construct a path from $A$ to $T$ and a path from $T$
to $B$, and then concatenate them to obtain a path from $A$ to $B$.
\end{enumerate}

\begin{definition}\rm
\label{outwardpathDEF} We call the path constructed above the {\em
outward path} between $A$ and $B$.
\end{definition}
Keil and Gutwin~\cite{keil}, from this point on, use
a purely geometric argument (with no use of Delaunay graph
properties) to show that the length of the obtained path $A = M_0,
M_1,\cdots, M_r = B$ (where each point $M_p$, for $p = 1, \cdots, r-1$, is
an intermediate point with respect to a pair $(M_i, M_j)$, where $0
\leq i < p < j \leq r$) is smaller than the length of
$\stackrel{\frown}{AB}$. Figure~\ref{outwardpath} illustrates an
outward path between $A$ and $B$.

\begin{figure}[htbp]

\begin{center} \small

\begin{pspicture}(-10,-10)(220,100)


\psarc[linecolor=darkgray, linewidth=0.3pt](100,-52){113}{20}{160}
\psarc[linewidth=0.5pt](0,0){25}{0}{35}

\cnode*(0,0){2pt}{c}\nput{-160}{c}{$C$}
\cnode*(200,0){2pt}{b}\nput{-20}{b}{$B = M_3$}
\cnode*(90,61){2pt}{a}\nput{90}{a}{$A= M_0$}

\cnode*(145,45){2pt}{m1}\nput[labelsep=8pt]{35}{m1}{$M_1$}
\cnode*(180,24){2pt}{m2}\nput{20}{m2}{$M_2$}

\ncline{c}{b} \ncline{b}{m2} \ncline{m2}{m1} \ncline{m1}{a}
\ncline{a}{c} \ncline{c}{m1} \ncline{c}{m2}

\put(0,40){\scriptsize $\leq 2\pi/k$}\pnode(10,40){a1}
\pnode(15,10){a2} \ncline[nodesepA=5,nodesepB=5]{->}{a1}{a2}

\end{pspicture} \caption{Illustration of an outward path.}\label{outwardpath}
\end{center}
\end{figure}
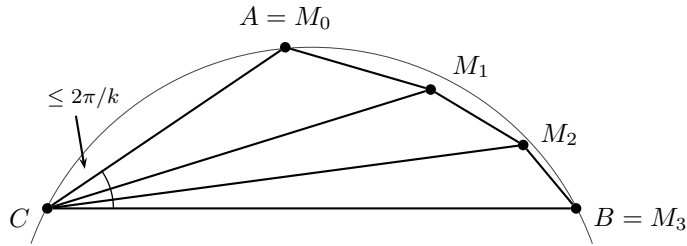

\begin{proposition}
\label{nocross} In every recursive step of the outward path
construction described above, if $M_p$ is an intermediate point with
respect to a pair of points $(M_i,M_j)$, then:
\begin{itemize}

\item[(a)] there is a circle passing through $C$ and $M_p$ that contains
no point of $G$, and

\item[(b)] circles $\bigcirc{CM_iM_p}$ and $\bigcirc{CM_jM_p}$ contain no
points of $G$ except, possibly, in the region subtended by chords $M_iM_p$
and $M_pM_j$, respectively, away from $C$.
\end{itemize}
\begin{proof}
We assume, by induction, that there are circles $(O_{M_i})$ and
$(O_{M_j})$ passing through $C$ and $M_i$, and $C$ and $M_j$,
respectively, containing no points of $G$, and that the circle
$(O)=\bigcirc{CM_iM_j}$ contains no point of $G$ in the interior of
the region $R'$ subtended by chord $M_iM_j$ closer to $C$.
(This is certainly true in the base case because $CA, CB \in G$, by
Proposition \ref{Euclidproperty}, and by our initial assumptions).

Since $M_iM_j$ is not an edge in $G$, the point $M_p$ chosen in the
construction is the point with the property that the region $R$ of
$\bigcirc{M_iM_pM_j}$ subtended by chord $M_iM_j$ away from $C$, contains
no point of $G$. Then the circle passing
through $C$ and $M_p$ and tangent to $\bigcirc{M_iM_pM_j}$ at $M_p$
is completely inside $(O_{M_i}) \cup (O_{M_j}) \cup R \cup R'$, and
therefore devoid of points of $G$. This proves part (a).

The region of $\bigcirc{CM_iM_p}$ subtended by chord
$M_iM_p$ and containing $C$ is inside $(O_{M_i}) \cup R \cup R'$,
and therefore contains no point of $G$ in its interior. The same is
true for the region of $\bigcirc{CM_jM_p}$ subtended by chord
$M_jM_p$ and containing $C$, and part (b) holds as well.
\end{proof}
\end{proposition}

We are now ready to prove Lemma~\ref{canpath} in the case when no
point of $G$ lies inside $\triangle{ABC}$. In this case we define
the path in Lemma~\ref{canpath} to be the outward path between $A$
and $B$.
\begin{proof}[Proof of Lemma~\ref{canpath} for the case of outward
path.]
\begin{itemize}
\item[]
\item[$(i)$] With $\Theta = \angle{BCA}$, we have $|\stackrel{\frown}{AB}| =
2 \Theta \cdot |OA|$ and $\sin{\Theta}= {|AB| / (2|OA|)}$. We note that
$|CA|+|\stackrel{\frown}{AB}|$ is largest
when $|CA| = |CB|$, i.e. when $CA$ and $CB$ are symmetrical with
respect to the diameter of $\bigcirc{CAB}$ passing through $C$; this
follows from the fact that the perimeter of a
convex body is not smaller than the perimeter of a convex body
containing it (see page 42 in~\cite{bookconvexity}). If $|CA| = |CB|$,
$\sin{\Theta \over 2} = {|AB| \over 2|CB|}$. Using elementary
trigonometry, it follows from the above facts and from
$|CA| \leq |CB|$ that:
\begin{eqnarray*}
|CA|+|\stackrel{\frown}{AB}| &   \leq & |CB|+2 \Theta \cdot |OA| =
|CB|+ ({\Theta \over \sin{\Theta}}) \cdot |AB| =  |CB| + ({\Theta
\over \cos{\Theta \over 2}}) \cdot |CB| \\
		 & \leq & (1+2\pi(k\cos{\frac{\pi}{k}})^{-1})|CB|.
\end{eqnarray*}
The last inequality follows from $\Theta \leq 2\pi/k$ and $k
> 2$.

\item[$(ii)$] If $M_iM_j$ was an edge in $G$ then, for every $p$ between $i$
and $j$, the circle $\bigcirc{M_iM_pM_j}$ would not contain $C$. This,
however, contradicts part (a) of Proposition~\ref{nocross}.

\item[$(iii)$] If the outward path contains a single intermediate point $M_1$,
then since $M_1$ lies inside $(O)=\bigcirc{CAB}$, $\angle{AM_1B}
\geq \pi - \angle{AOB}/2 \geq \pi - 2\pi/k = (k-2)\pi/k$ (note that
$\angle{AOB} = 2 \cdot \angle{ACB}$), as desired. Now the statement
follows by induction on the number of steps taken to construct the
outward path between $A$ and $B$, using the fact (proved
in~\cite{keil}) that each angle $\angle{M_{i-1}O_iM_{i+1}}$ at the
center of the circle $(O_i)$ defining the intermediate point $M_i$,
is bounded by $\angle{AOB}$.
\item[$(iv)$] This follows from the fact that $\angle{CAM_1} \geq
\angle{CAB} \geq \pi/2 - \pi/k$. The last inequality is true
because $|CA| \leq |CB|$ and $\angle{BCA} \leq 2\pi/k$ in
$\triangle{CAB}$.
\end{itemize}
\end{proof}

\subsection{The Inward Path}
\label{inwardpatheuclidean} We consider now the case when the
interior of $\triangle{ABC}$ contains points of $G$. Let $S$ be the
set of points consisting of points $A$ and $B$ plus all the points
interior to $\triangle{ABC}$ (note that $C \notin S$). Let $CH(S)$
be the points on the convex hull of $S$. Then $CH(S)$ consists of
points $N_0=A$ and $N_s = B$, and points $N_1, ..., N_{s-1}$ of $G$
interior to $\triangle{ABC}$. We have the following proposition:

\begin{proposition}
\label{inwardproperties} For every $i =1, \cdots, s-1:$
\begin{itemize}
\item[(a)] $CN_i \in G$,
\item[(b)] $|CN_i| \leq |CN_{i+1}|$, and
\item[(c)] $\angle{N_{i-1}N_iN_{i+1}} \geq \pi$, where $\angle{N_{i-1}N_iN_{i+1}}$ is the angle facing point
$C$.
\end{itemize}
\begin{proof}
These follow follow from the following facts: $CA$ and $CB$ are edges in $G$,
$CA$ is the shortest edge in its cone, and hence $|CA| \leq |CN_i|$, for $i=0,
\cdots, s$, and points $N_0, \cdots, N_s$ are on $CH(S)$ in the listed order.
\end{proof}
\end{proposition}

Since $|CN_i| \leq |CN_{i+1}|$ and no point of $G$ lies inside
$\triangle{N_iCN_{i+1}}$ ($N_i$ and $N_{i+1}$ are on $CH(S)$),
$CN_i$ is the shortest edge in the angular sector
$\angle{N_iCN_{i+1}}$. Since $\angle{N_iCN_{i+1}} \leq
\angle{BCA}\leq 2\pi/k$, by Lemma~\ref{canpath} there exists an
outward path $P_i$ between $N_i$ and $N_{i+1}$, for every $i = 0, 1,
\cdots, s-1$, satisfying all the properties of Lemma~\ref{canpath}.
Let $A=M_0, M_1, \cdots, M_r=B$ be the concatenation of the paths
$P_i$, for $i=0, \cdots, r-1$.

\begin{definition}\rm
\label{inwardpathDEF} We call the path $A=M_0, M_1, \cdots, M_r=B$
constructed above the {\em inward path} between $A$ and $B$.
\end{definition}
Figure~\ref{inwardpath} illustrates an inward path between $A$ and
$B$.

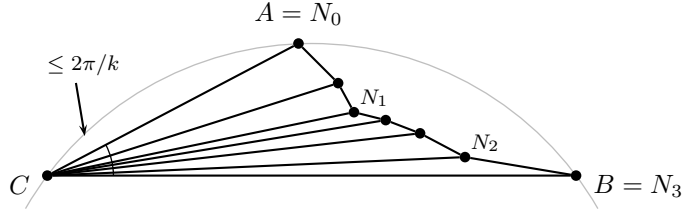
\begin{figure}[htbp]
\begin{center} \small
\begin{pspicture}(-10,-10)(220,80)

\psarc[linecolor=lightgray, linewidth=0.5pt](100,-75){125}{30}{150}
\psarc[linewidth=0.5pt](0,0){25}{0}{28}

\cnode*(0,0){2pt}{c}\nput{-160}{c}{$C$}
\cnode*(200,0){2pt}{b}\nput{-20}{b}{$B = N_3$}
\cnode*(95,50){2pt}{a}\nput{90}{a}{$A= N_0$}

\cnode*(110,35){2pt}{m1}
\cnode*(116,24){2pt}{m2}\nput[labelsep=1pt]{45}{m2}{\scriptsize
$N_1$}
\cnode*(128,21){2pt}{m3}
\cnode*(141,16){2pt}{m4}
\cnode*(158,7){2pt}{m5}\nput[labelsep=1pt]{45}{m5}{\scriptsize
$N_2$}


\ncline{c}{b} \ncline{c}{m1} \ncline{c}{m2} \ncline{c}{m3}
\ncline{c}{m4} \ncline{c}{m5} \ncline{b}{m5} \ncline{m5}{m4}
\ncline{m4}{m3} \ncline{m3}{m2} \ncline{m2}{m1} \ncline{m1}{a}
\ncline{a}{c}

\put(0,40){\scriptsize $\leq 2\pi/k$}\pnode(10,40){a1}
\pnode(15,10){a2} \ncline[nodesepA=5,nodesepB=5]{->}{a1}{a2}
\end{pspicture}
\caption{Illustration of an inward path.}\label{inwardpath}
\end{center}
\end{figure}

We now prove Lemma~\ref{canpath} in the case when there are points
of $G$ interior to $\triangle{ABC}$. In this case we define the path
in Lemma~\ref{canpath} to be the inward path between $A$ and $B$.

\begin{proof}[Proof of Lemma~\ref{canpath} for the case of inward
path.]
\begin{itemize}
\item[]
\item[$(i)$]

Define $A''$ to be a point on the half-line $[CA$ such that $|CA''|
= |CB|$, and let $(O'') =\bigcirc{CA''B}$. Denote by $\alpha''$ the
length of the arc of $\bigcirc{CA''B}$ subtended by chord $A''B$ and
facing $\angle{A''CB}$. For every $i = 0, 1, \cdots, s-1$, we define
arc $\alpha_i$ to be the arc of $\bigcirc{CN_iN_{i+1}}$ subtended by
chord $N_iN_{i+1}$ and facing $\angle{N_iCN_{i+1}}$. For every $i =
0, 1, ..., s-1$, we define $N_i'$ to be the point on the half-line
$[CN_i$ such that $|CN_i'| = |CN_{i+1}|$, $(O_i)$ to be the circle
$\bigcirc{CN_i'N_{i+1}}$, and $\alpha'_{i}$ to be the arc of $(O_i)$
subtended by chord $N_i'N_{i+1}$ and facing $\angle{N_i'CN_{i+1}}$.
Finally, for every $i =0, \cdots, s-1$, we define $N_i''$ to be the
point of intersection of the half-line $[CN_i$ and circle $(O'')$,
and $\alpha''_i$ to be the arc of $(O'')$ subtended by chord
$N_i''N''_{i+1}$ and facing $\angle{N_i''CN''_{i+1}}$. As shown in
section~\ref{outpath}, the length of the outward path $P_i$ between
$N_i$ and $N_{i+1}$ is bounded by the length of $\alpha_i$. Since
the convex body $C_1$ delimited by $CN_i$, $CN_{i+1}$ and $\alpha_i$
is contained inside the convex body $C_2$ delimited by $CN_i'$,
$CN_{i+1}$ and $\alpha'_i$, by~\cite{bookconvexity}, the perimeter
of $C_1$ is not larger than that of $C_2$. Denoting by $|P_i|$ the
length of path $P_i$, we get:
\begin{eqnarray}
|P_i| \leq |N_iN'_{i}| + \alpha'_i, \ \ \ \ i = 1, \cdots, s-1.
\label{eq1}
\end{eqnarray}
Since $(O_i)$ and $(O'')$ are concentric circles (of center $C)$,
and the radius of $(O_i)$ is not larger than that of $(O'')$, we
have $\alpha'_i \leq \alpha''_i$, for $i=0, \cdots, s-1$. It follows
from Inequality~(\ref{eq1}) that:
\begin{eqnarray}
|P_i| \leq |N_iN'_{i}| + \alpha''_i, \ \ \ \ i = 1, \cdots, s-1.
\label{eq2}
\end{eqnarray}
Using Inequalities~(\ref{eq1}) and (\ref{eq2}) we get:
\begin{eqnarray}
\hspace*{-2mm} |CA| + \sum_{i=0}^{s-1}|P_i| \leq |CA| +
\sum_{i=0}^{s-1} |N_iN'_{i}| + \sum_{i=0}^{s-1} \alpha''_i.
\hspace*{-2mm}\label{eq3}
\end{eqnarray}
Noting that $\sum_{i=0}^{s-1} |N_iN'_{i}| = |CB| - |CA|$, that
$\sum_{i=0}^{r-1} \alpha''_i = \alpha''$, and using the same argument as in
part ($i$) of Lemma~\ref{canpath}) completes the proof.

\item[$(ii)$] Since $CN_p \in G$ for $p=1, \cdots, s-1$ by part $(a)$ of
Proposition~\ref{inwardproperties}, by planarity of $G$, if such an
edge between two points $M_i$ and $M_j$ exists, then $M_i$ and $M_j$
must belong to an outward path between two points $N_p$ and
$N_{p+1}$ of $CH(S)$. But this contradicts part $(ii)$ of
Lemma~\ref{canpath} for the case of the outward path applied to
$N_p$ and $N_{p+1}$.

\item[$(iii)$] For each $i= 0, \cdots, r$, either $M_i =N_j \in CH(S)$,
or $M_i$ is an intermediate point on the outward path between two
points $N_p$ and $N_q$ in $CH(S)$. In the former case
$\angle{M_{i-1}M_iM_{i+1}} \geq \angle{N_{j-1}M_iN_{j+1}} \geq \pi
\geq (k-2)\pi/k$ for $k\geq 14$ ($N_{j-1}$ and $N_j$ are points
before and after $M_i=N_j$ on $CH(S)$), by part (c) of
Proposition~\ref{inwardproperties}. In the latter case
$\angle{M_{i-1}M_iM_{i+1}} \geq (k-2)\pi/k$ by the proof of part
$(iii)$ of Lemma~\ref{canpath} applied to the outward path between
$N_p$ and $N_q$.

\item[$(iv)$] This follows from $|CA|=|CM_0| \leq |CM_1|$ and
$\angle{ACM_1} \leq \angle{ACB} \leq 2\pi/k$, in triangle
$\triangle{CAM_1}$.
\end{itemize}
\end{proof}

\subsection{The Modified Yao Step}
\label{modifiedyaoalgo}

We now augment the {\em Yao step} so edges forming the paths described in
Lemma~\ref{canpath} are included in $G'$, in addition to the edges chosen
in the standard Yao step. Lemma~\ref{canpath} says that consecutive edges
on such paths form moderately large angles. The modified Yao step will
ensure that consecutive edges forming large angles are included in $G'$.
The algorithm is described in Figure~\ref{yaostep}.

\begin{algorthm}{The modified Yao Step.}{yaostep} \algtitle{Modified Yao step}
\alginput{A Delaunay graph $G$; integer $k \geq 14$} \algoutput{A
subgraph $G'$ of $G$ of maximum degree $k$}
\begin{codeblock}

\step define $k$ disjoint cones of size $2\pi/k$ around every point
$M$ in $G$;

\step  in every non-empty cone, select the shortest edge incident on
$M$ in this cone;

\step {\bf for} every maximal sequence of $\ell \geq 1$ consecutive
empty cones:

\hspace*{3mm} 3.1. {\bf if} $\ell > 1$ {\bf then} select the first
$\lfloor \ell/2 \rfloor$ unselected  \\
\hspace*{7.9mm}   incident edges on $M$ clockwise from the sequence
\\ \hspace*{7.9mm} of empty cones and the first $\lceil \ell/2 \rceil$ unselected  \\ \hspace*{7.9mm} edges
incident on $M$ counterclockwise from the \\
\hspace*{7.9mm} sequence  of empty cones; \\
\hspace*{3mm} 3.2. {\bf else} (i.e., $\ell = 1$) let $MX$ and $MY$
be the incident  \\ \hspace*{7.9mm} edges on $M$ clockwise and
counterclockwise,  \\ \hspace*{7.9mm} respectively, from the empty
cone; {\bf if} either $MX$ \\ \hspace*{7.9mm} or $MY$ is selected
{\bf then}  select the other edge  \\
\hspace*{7.9mm} (in case it has not been selected); {\bf otherwise} select\\
\hspace*{7.9mm}  the shorter edge between $MX$ and
 $MY$ breaking \\ \hspace*{7.9mm} ties arbitrarily;

\step $G'$ is the spanning subgraph of $G$ consisting of edges selected by
both endpoints.
\end{codeblock}
\end{algorthm}

Since the algorithm selects at most $k$ edges incident on any point
$M$ and since only edges chosen by both endpoints are included in $G'$,
each point has degree at most $k$ in $G'$.

Before we complete the proof of Theorem~\ref{maintheorem}, we show that the
running time of the algorithm is linear. Note first that all edges incident
on point $M$ of degree $\Delta$ can be mapped to the $k$ cones around $M$ in
linear time in $\Delta$. Then, the shortest edge in every cone can be found in
time $O(\Delta)$ (step 2. in the algorithm). Since $k$ is a constant, selecting
the $\ell/2$ edges clockwise (or counterclockwise) from a sequence of $ell < k$
empty cones around $M$ (step 3.1.) can be done in $O(\Delta)$ time. Noting that
the total number of edges in $G$ is linear in the number of vertices completes
the analysis.

To complete the proof of Theorem~\ref{maintheorem}, all we need to do is show:
\begin{lemma}
\label{edgeselected}
If edge $CB$ is not selected by the algorithm, let $CA$ be the shortest edge in
the cone out of $C$ to which $CB$ belongs. Then the edges of the path
described in Lemma~\ref{canpath} are included in $G'$  by the algorithm.
\begin{proof}
For brevity, instead of saying that the algorithm {\bf Modified Yao
Step} selects an edge $MX$ out of a point $M$, we will say that $M$
selects edge $MX$. To get started, it is obvious that $C$ will select edge
$CA$.

By part $(iv)$ of Lemma~\ref{canpath}, the angle $\angle{CAM_1} \geq
\pi/ 2 - \pi/k \geq 6 \pi/k$ for $k \geq 14$. Therefore, at least
two empty cones must fall within the sector $\angle{CAM_1}$
determined by the two consecutive edges $CA$ and $AM_1$, and edges
$AC$ and $AM_1$ will both be selected by $A$. Since edge $CA$ is
also selected by point $C$, edge $AC \in G'$.

By part $(iii)$ of Lemma~\ref{canpath}, for every $i = 1, 2, \cdots,
r-1$, the angle $\angle{M_{i-1}M_iM_{i+1}} \geq (k-2)\pi/k \geq
10\pi/k$ for $k\geq 12$, and hence at least four cones fall within
the angular sector $\angle{M_{i-1}M_iM_{i+1}}$. Since by part {\em
(ii)} of Lemma~\ref{canpath} $M_iC$ is the only possible edge inside
the angular sector $\angle{M_{i-1}M_iM_{i+1}}$, it is easy to see
that regardless of the position of these four cones with respect to
edge $M_iC$, $M_i$ ends up selecting all edges $M_iM_{i-1}$,
$M_iM_{i+1}$ and $M_iC$ in steps 2 and/or 3 of the algorithm. Since
we showed above that $A$ selects edge $AM_1$, this shows that all
edges $M_iM_{i+1}$, for $i=0, \cdots, r-2$, are selected by both
their endpoints, and hence must be in $G'$. Moreover, edge
$M_{r-1}M_r=M_{r-1}B$ is selected by point $M_{r-1}$.

We now argue that edge $BM_{r-1}$ will be selected by $B$. First, observe
that $|BM_{r-1}| \leq |\stackrel{\frown}{AB}| < |CB|$. 
Let $CD$ be the other consecutive edge to $CB$ in $G$ (other than
$CM_{r-1}$). Because $C$ does not select $B$, it follows that
$\angle{M_{r-1}CD} \leq 6 \pi/k$. Otherwise, since $CM_{r-1}$ and
$CB$ are in the same cone, two empty cones would fall within the
sector $\angle{BCD}$ and $C$ would select $B$. Since $CB$ is an edge
in $G$, by the characterization of Delaunay edges~\cite{book},
$\angle{CM_{r-1}B} + \angle{CDB} \leq \pi$. By considering the
quadrilateral $CDBM_{r-1}$, we have $\angle{M_{r-1}CD} +
\angle{DBM_{r-1}} \geq \pi$. This, together with the fact that
$\angle{M_{r-1}CD} \leq 6 \pi/k$, imply that $\angle{DBM_{r-1}} \geq
(k-6)\pi/k \geq 8 \pi/k$, for $k\geq 14$. Therefore,
$\angle{DBM_{r-1}}$ contains at least three cones of size $2\pi/k$
out of $B$. If one of these cones falls within the angular sector
$\angle{CBM_{r-1}}$ then, since $|M_{r-1}B| < |CB|$, $BM_{r-1}$ must
have been selected out of $B$.

Suppose now that $\angle{CBM_{r-1}}$ contains no cone inside and
hence $\angle{CBM_{r-1}} < 4 \pi/k$. If one of these three cones
within sector $\angle{DBM_{r-1}}$ contains edge $CB$, then the
remaining two cones must fall within $\angle{DBC}$ and $BM_{r-1}$
will get selected out of $B$ when considering the sequence of at
least two empty cones contained within $\angle{CBD}$. Suppose now
that all three empty cones fall within $\angle{CBD}$. Then we have
$\angle{CBD} \geq 6\pi/k$.

If $\angle{M_{r-1}CD} \geq 4 \pi/k$, then since $M_{r-1}C$ and $CB$
belong to the same cone, the sector $\angle{BCD}$ must contain an
empty cone. Because $D$ is exterior to $\bigcirc{CBM_{r-1}}$,
$\angle{CBM_{r-1}} < 4 \pi/k$, and $\angle{M_{r-1}CB} \leq 2\pi/k$,
it follows that $\angle{CDB} < \angle{M_{r-1}CB} + \angle{CBM_{r-1}}
< 6 \pi /k < \angle{DBC}$. Therefore, by considering the triangle
$\triangle{CDB}$, we note that $|CB| < |CD|$. But then edge $CB$
would have been selected by $C$ in step 3 since the sector
$\angle{BCD}$ contains an empty cone, a contradiction.

It follows that $\angle{M_{r-1}CD} \leq 4 \pi/k$, and therefore
$\angle{M_{r-1}BD} \geq (k-4) \pi/k \geq 10 \pi/k$ for $k \geq 14$.
This means that at least four cones are contained inside sector
$\angle{DBM_{r-1}}$. It is easy to check now that regardless of the
placement of the edge $BC$ with respect to these cones, edge
$BM_{r-1}$ is always selected out of $B$ by the algorithm. This
completes the proof.
\end{proof}
\end{lemma}

\begin{corollary}
A Euclidean Minimum Spanning Tree (EMST) on $P$ is a subgraph of $G'$. 
\begin{proof}
It is well known that a Delaunay graph ($G$) contains a EMST. If an edge $CB$
is not in $G'$, then, by Lemma~\ref{edgeselected}, a path from $C$ to $B$ is
included in $G'$. All edges on this path are no longer than $CB$, so there is
a EMST not including $CB$. 
\end{proof}
\end{corollary}
Since a Delaunay graph of a Euclidean graph of $n$ points can be
computed in time $O(n\lg{n})$~\cite{book} and has stretch factor
$C_{del}\approx 2.42$, we have the following theorem.
\begin{theorem}
\label{spannereuclidean} There exists an algorithm that, given a
set $P$ of $n$ points in the plane, computes a plane geometric spanner
of the Euclidean graph on $P$ that contains a EMST, has maximum degree $k$,
and has stretch factor
$(1+2\pi(k\cos{\frac{\pi}{k}})^{-1}) \cdot C_{del}$, where $k \geq 14$
is an integer. Moreover, the algorithm runs in time $O(n\lg{n})$.
\end{theorem}

\section{Geometric Spanners of Unit Disk Graphs}
  \label{structresults}  In this section we generalize our planar geometric
spanner algorithm to unit disk graphs. Unit disk graphs model wireless ad-hoc
and sensor networks and, for packet routing and other applications, a
bounded-degree planar geometric spanner of the wireless network is often
desired. Due to the limited computational power of the network devices and
the requirement that the network be robust with respect to device joining
and leaving the network, the construction/algorithm should ideally be
{\em strictly-localized}: the computation performed at a point depends solely
on the information available at the point and its $d$-hop neighbors, for some
constant $d$ (in our case $d=2$). In particular, no global propagation of
information should take place in the network.

The results in the previous section do not carry over to unit disk graphs
because not all Delaunay graph edges on a point-set $P$ are
unit disk edges. However, if $U$ is the unit disk graph on points in $P$ and
$UDel(U)$ is the subgraph of the Delaunay graph on $P$ obtained by deleting
edges of length greater than one unit, then $UDel(U)$ is a connected,
planar, spanning subgraph of $U$ with stretch factor bounded by $C_{del}$
(see~\cite{iitunbounded1, BoseMNSZ04}). Therefore, if we apply the results
from the previous section to
$UDel(U)$ and observe that all edges on the path defined in
Lemma~\ref{canpath} must be unit disk edges (given that edges $CA$
and $CB$ are), it is easy to see that Theorem~\ref{maintheorem} and
Theorem~\ref{spannereuclidean} carry over to unit disk graphs. The
only problem, however, is that the construction of $UDel(U)$ cannot
be done in a strictly-localized manner.

To solve this problem, Wang et al.~\cite{iitunbounded1,iitunbounded}
introduced a subgraph of $U$ denoted $LDel^{(2)}(U)$. It was shown
in~\cite{iitunbounded1,iitunbounded} that $LDel^{(2)}(U)$ is a
planar supergraph of $UDel(U)$, and hence also has stretch factor bounded
by $C_{del}$. Moreover, the results in~\cite{gruia,iitbounded} show
how $LDel^{(2)}(U)$ can be computed with a strictly-localized
distributed algorithm exchanging no more than $O(n)$ messages in
total ($n$ is the number of points in $U$), and having a local
processing time of $O(\Delta\lg{\Delta}) = O(n\lg{n})$ at a point of
degree $\Delta$. In a style similar to Definition~\ref{DelaunayDef},
$LDel^{(2)}(U)$ can be defined as follows:
\begin{definition}
\label{ldel2}
An edge $XY$ of $U$ is in $LDel^{(2)}(U)$ if and only if there exists a
circle through points $X$ and $Y$ whose interior contains no point of $U$
that is a 2-hop neighbor of $X$ or $Y$.
\end{definition}
We will use $G=LDel^{(2)}(U)$ as the underlying subgraph
of $U$ to replace the Delaunay graph $G$ used in the previous
section. We note that $G$ is planar, is a
supergraph of $UDel(U)$, and hence has stretch factor $C_{del}$. To
translate our results to unit disk graphs, we need to show that the
inward and outward paths are still well defined in $G$.
In particular, we need to show that Lemma~\ref{canpath} holds for
$G=LDel^{(2)}(U)$. We outline the general approach and omit the
details for lack of space.

The following is equivalent to Proposition~\ref{Euclidproperty}:
\begin{lemma}
\label{UnitDiskproperty} If $CA$ and $CB$ are edges of $G$ then the region
of $(O) = \bigcirc{ABC}$ subtended by chord $CA$ and away from $B$ and
the region of $(O)$ subtended by chord $CB$ and away from $A$ contain no
points that are two hop neighbors of $A$, $B$ and $C$.
\begin{proof}
By symmetry it is enough to prove the lemma for  the region of $(O)$
subtended by chord $CA$ and away from $B$. By Definition~\ref{ldel2},
there is a circle $(O_{CA})$ passing through $C$ and $A$ whose interior
is empty of any point within two hops of $C$ or $A$. The region of $(O)$
subtended by chord $CA$ and away from $B$ is inside this circle, so we only
need to argue that it doesn't contain two hop neighbors of $B$ either. If it
did, say point $X$, then any neighbor of $X$ and $B$ would have to be a
neighbor of $C$ or $A$ as well, a contradiction.
\end{proof}
\end{lemma}
With this lemma in hand, the recursive construction of the outward path given
in Subsection~\ref{outpath} can be applied to the graph $G = LDel^{(2)}(U)$.
The following proposition for $G = LDel^{(2)}(U)$ corresponds to
Proposition~\ref{nocross} for Delaunay graphs and is proven in an equivalent
manner:
\begin{proposition}
\label{nocross2} In every recursive step of the outward path
construction, if $M_p$ is an intermediate point with
respect to a pair of points $(M_i,M_j)$, then:
\begin{itemize}

\item[(a)] there is a circle passing through $C$ and $M_p$ that contains
no point of $G$ that is a two-hop neighbor of $C$ or $M_p$, and

\item[(b)] circles $\bigcirc{CM_iM_p}$ and $\bigcirc{CM_jM_p}$ contain no
points of $G$ that are two-hop neighbors of $C$, $M_i$ and $M_p$ and $C$,
$M_j$, and $M_p$, respectively, except, possibly, in the region subtended by
chords $M_iM_p$ and $M_pM_j$, respectively, away from $C$.
\end{itemize}
\end{proposition}
With this proposition, we can show that Lemma~\ref{canpath} holds
true for $G = LDel^{(2)}(U)$ for outward paths. It holds for inward paths
as well, using the same argument as in Section~\ref{inwardpatheuclidean}.  
Finally, it is obvious how the {\bf Modified Yao Step} algorithm in
Section~\ref{modifiedyaoalgo} can be easily described as a
strictly-localized algorithm. We can show, therefore, the following theorem:
\begin{theorem}
\label{spannerunit} There exists a distributed strictly-localized
algorithm that, given a set $P$ of $n$ points in the plane, computes a plane
geometric spanner of the unit disk graph on $P$ that contains a EMST, has
maximum degree $k$, and has stretch factor
$(1+2\pi(k\cos{{\frac{\pi}{k}}})^{-1})\cdot C_{del}$,
for any integer $k \geq 14$. Moreover, the algorithm exchanges no
more than $O(n)$ messages in total, and has a local processing time
of $\Delta\lg{\Delta}$ at a point of degree $\Delta$.
\end{theorem}
Due to the strictly-localized nature of the algorithm, the algorithm
is very robust to topological changes (such as wireless devices moving or
joining or leaving the network), an essential property for
the application of the algorithm in a wireless ad-hoc environment.
\bibliographystyle{plain}

\begin{thebibliography}{10}

\bibitem{bookconvexity}
R.~Benson.
\newblock {\em Euclidean Geometry and Convexity}.
\newblock Mc-Graw Hill, New York, 1966.

\bibitem{boseesa}
P.~Bose, J.~Gudmundsson, and M.~Smid.
\newblock Constructing plane spanners of bounded degree and low weight.
\newblock In {\em Proceedings of the 10th Annual European Symposium on
  Algorithms}, volume 2461 of {\em Lecture Notes in Computer Science}, pages
  234--246. Springer, 2002.

\bibitem{bosealgorithmica}
P.~Bose, J.~Gudmundsson, and M.~Smid.
\newblock Constructing plane spanners of bounded degree and low weight.
\newblock {\em Algorithmica}, 42(3-4):249--264, 2005.

\bibitem{BoseMNSZ04}
P.~Bose, A.~Maheshwari, G.~Narasimhan, M.~Smid, and N.~Zeh.
\newblock Approximating geometric bottleneck shortest paths.
\newblock {\em Computational Geometry: Theory and Applications}, 29:233--249,
  2004.

\bibitem{bose1}
P.~Bose, M.~Smid, and D.~Xu.
\newblock Diamond triangulations contain spanners of bounded degree.
\newblock {\em To appear in Algorithmica}, 2007.

\bibitem{gruia}
G.~Calinescu.
\newblock Computing 2-hop neighborhoods in {Ad Hoc} wireless networks.
\newblock In {\em Proceedingsof the 2nd International Conference on Ad-Hoc,
  Mobile, and Wireless Networks}, volume 2865 of {\em Lecture Notes in Computer
  Science}, pages 175--186. Springer, 2003.

\bibitem{chew}
P.~Chew.
\newblock There are planar graphs almost as good as the complete graph.
\newblock {\em Journal of Computers and System Sciences.}, 39(2):205--219,
  1989.

\bibitem{book}
M.~de~Berg, M.~van Kreveld, M.~Overmars, and O.~Schwarzkopf.
\newblock {\em Computational Geometry: Algorithms and Applications}.
\newblock Springer-Verlag, second edition, 2000.

\bibitem{dobkin}
D.~Dobkin, S.~Friedman, and K.~Supowit.
\newblock Delaunay graphs are almost as good as complete graphs.
\newblock {\em Discrete Computational Geometry}, 5(4):399--407, 1990.

\bibitem{keil}
J.~Keil and C.~Gutwin.
\newblock Classes of graphs which approximate the complete {Euclidean} graph.
\newblock {\em Discrete {\&} Computational Geometry}, 7:13--28, 1992.

\bibitem{iitunbounded1}
X.-Y. Li, G.~Calinescu, and P.-J. Wan.
\newblock Distributed construction of planar spanner and routing for ad hoc
  wireless networks.
\newblock In {\em Proceedings of the {IEEE} INFOCOM}, 2002.

\bibitem{iitunbounded}
X.-Y. Li, G.~Calinescu, P.-J. Wan, and Y.~Wang.
\newblock Localized delaunay triangulation with application in {Ad Hoc}
  wireless networks.
\newblock {\em IEEE Transactions on Parallel and Distributed Systems.},
  14(10):1035--1047, 2003.

\bibitem{spannerbook}
G.~Narasimhan and M.~Smid.
\newblock {\em Geometric Spanner Networks}.
\newblock Cambridge University Press, 2007.

\bibitem{iitbounded1}
Y.~Wang and X.-Y. Li.
\newblock Localized construction of bounded degree and planar spanner for
  wireless ad hoc networks.
\newblock In {\em Proceedings of the {DIALM-POMC} Joint Workshop on Foundations
  of Mobile Computing}, pages 59--68. ACM, 2003.

\bibitem{iitbounded}
Y.~Wang and X.-Y. Li.
\newblock Localized construction of bounded degree and planar spanner for
  wireless ad hoc networks.
\newblock {\em MONET}, 11(2):161--175, 2006.

\bibitem{yao}
A.~C.-C. Yao.
\newblock On constructing minimum spanning trees in k-dimensional spaces and
  related problems.
\newblock {\em SIAM Journal on Computing}, 11(4):721--736, 1982.

\end{thebibliography}

\end{document}